\documentclass[journal,comsoc,twocolumns,10pt]{IEEEtran}

\usepackage[T1]{fontenc}
\usepackage{bm}
\def\bs{\bm}
\usepackage{amssymb}
\usepackage{amsmath}

\usepackage{amsthm}
\newtheorem{theorem}{Theorem}[section]

\newtheorem{proposition}{Proposition}
\newtheorem{corollary}{Corollary}

\usepackage[noadjust]{cite}

\usepackage[cmintegrals]{newtxmath}
\DeclareSymbolFont{largesymbolsCM}{OMX}{cmex}{m}{n}

\let\sum\relax
\DeclareMathSymbol{\sum}{\mathop}{largesymbolsCM}{"50}

\DeclareMathAlphabet{\mathcal}{OMS}{cmsy}{m}{n}

\usepackage{newtxtext}

\newcommand{\Mathscr}[1]{\mathcal{#1}}
\newcommand{\Mathsf}[1]{\text{\usefont{OT1}{cmss}{m}{n}#1}}

\usepackage{multirow}
\usepackage[shrink=40, stretch=20]{microtype}
\renewcommand{\geq}{\geqslant}
\renewcommand{\leq}{\leqslant}
\def\be{\begin{equation}}
\def\ee{\end{equation}}
\usepackage{stackengine}
\setstackgap{S}{-1.5pt}
\def\leqsim{\lesssim}
\def\geqsim{\gtrsim}
\def\+{\hspace{0.2ex}}
\def\-{\hspace{-0.2ex}}

\usepackage[dvipsnames, table]{xcolor}
\usepackage{tikz-cd}
\tikzcdset{arrow style=tikz, diagrams={>=stealth}}

\def\sub#1{_{\mathrm{#1}}}
\newcommand{\E}[2]{\mathbf{E}_{#1}[\+{#2}\+]}
\newcommand{\Epar}[2]{\mathbf{E}_{#1}\!\left[{#2}\right]}
\def\T{\text{\!$^{(T)}$}}

\allowdisplaybreaks

\begin{document}
\title{Approaching Capacity Without Pilots \\ via Nonlinear Processing at the Edge}

\author{Guido Carlo Ferrante,~\IEEEmembership{Member,~IEEE}
\thanks{Ericsson Research, Stockholm, Sweden.}
}

\markboth{}%
{}

\maketitle
\thispagestyle{empty}

\begin{abstract}
A nonlinear detector derived within a maximum likelihood estimation framework is shown to be effective in retrieving the channel coefficients and data of users on the uplink channel of a noncooperative wireless system without the access point having any prior channel state information (no CSI or noncoherent setup). Rather than relying on pilot-assisted transmissions, it is shown that a maximum likelihood-based detector emerges naturally from an information-theoretic argument. The assumptions under which the detector is designed are as follows: 1) the uplink data from different users are independent and non-Gaussian; 2) the coherence block of the channel is much larger than the number of users (in practice, the square of the number of users); 3) the number of antennas at the access point or base station is equal to the number of users; 4) users continuously transmit within the coherence block; and 5) the transmission occurs at high signal-to-noise ratio. No coordination between the access point and unintended users (interference) is needed. Some coordination with intended users is needed. Finally, the system is assumed to be symbol-synchronous. 
\end{abstract}

\section{Introduction}

The problem tackled in this paper is that of estimating the data transmitted by the users inside a cell in the uplink of a cellular network without using pilots. Removing pilots from the start would also remove the possibility of pilot contamination, which has been a theme of considerable research \cite{Mar:2010, MulCotVeh:2014, HoytenDeb:2013}. 
Such detection without pilots can be classified as blind detection \cite{Honetal:1995, FerGerQueWin:2017}. We start from first principles and interpret the problem of estimating the channels between the users and the base station through the lens of maximum likelihood estimation (MLE). 

The results of this paper are two: first, MLE is shown to be effective in estimating channels and data of all users in the network provided that the coherence block size for these users is large enough; second, the MLE is emergent from an information-theoretic argument, which makes it a natural, and asymptotically optimal, detector.


There are three parameters that play a vital role in determining the performance of the network at high signal-to-noise ratio (SNR): the number of users in the network, $K$; the length of the coherence block, $T$; and the number of antennas at each base station, $n$. When $T>2K$, the possible presence of pilot contamination (due to lack of coordination among different cells) destroys some of the degrees of freedom of the uplink channel. In fact, the rate of all users that used the same pilot sequence is upper bounded by a constant at high SNR \cite{Mar:2010, HoytenDeb:2013}. This is not the behavior that we would get by using a centralized approach where pilots are not reused in different cells; in that case, all the degrees of freedom of the channel could be exploited. We show that it is possible to achieve the latter without pilots and with no coordination among cells. 



\subsection*{Notation} It is useful to denote by $\geqsim$ (resp. $\leqsim$) inequalities that are true up to a constant in $\rho$. For example, if $\varphi_{1}$ and $\varphi_{2}$ are functions of $\rho$, then we can write $\varphi_{1}(\rho)\geqsim \varphi_{2}(\rho)$ as a shorthand for $\varphi_{1}(\rho)\geq \varphi_{2}(\rho) + O_{\rho}(1)$, where $O_{\rho}(1)$ is a constant in $\rho$. We write $\varphi_{1}\simeq \varphi_{2}$ if $\varphi_{1}\geqsim \varphi_{2}$ and $\varphi_{1}\leqsim \varphi_{2}$. Moreover, we use $x\wedge y$ as a shorthand for $\min\{x,y\}$. The $i$\textsuperscript{th} vector of a canonical basis is denoted by $\bs{e}_{i}$, that is, $(\bs{e}_{i})_{j}=\delta_{ij}$ where $\delta_{ij}$ denotes the Kronecker delta; the dimension of the vector is kept implicit.

\section{System Model}
Suppose $K$ users populate a network with $L$ base stations, each of which is equipped with $n$ antennas. A block-fading model is assumed, where the coherence block size is equal to $T$ channel uses. The signal received by base station $\ell$ over the coherence block can be written as follows:
\begin{align} 
\bs{Y}_{\ell} = \sum_{\ell'=1}^{L} \bs{H}_{\ell'\ell}\bs{X}_{\ell'} + \bs{Z}_{\ell} 
\end{align} 
where $\bs{Y}_{\ell}$ is the $n\times T$ matrix of received symbols, $\bs{H}_{\ell'\ell}$ is the $n\times K_{\ell'}$ channel between the $K_{\ell'}$ users in cell $\ell'$ and the base station, $\bs{X}_{\ell'}$ is the $K_{\ell'}\times T$ matrix of symbols transmitted by users in cell $\ell'$, and $\bs{Z}_{\ell}$ is an additive white Gaussian noise. Without loss of generality, the elements of $\bs{Z}_{\ell}$ are distributed according to a proper complex Normal distribution with zero mean and unit variance: $(\bs{Z}_{\ell})_{ij}\sim\mathcal{CN}(0,1)$. We assume $\E{}{\| \bs{x}_{\ell k}\|^{2}} \leq T\rho$, where $\bs{x}_{\ell k}$ is the $k$\textsuperscript{th} row of $\bs{X}_{\ell}$, that is, the power constraint is assumed on the symbols transmitted by each user over the $T$ channel uses of the coherence block. 

For compactness, $\bs{Y}_{\ell}$ can be rewritten in the equivalent form
\begin{align} 
\bs{Y}_{\ell} & = \bs{H}_{\ell\ell} \bs{X}_{\ell} + \bs{H}_{-\ell,\ell} \bs{X}_{-\ell} + \bs{Z}_{\ell} \\ 
&  = \bs{H}_{\ell}\bs{X}+\bs{Z}_{\ell} 
\end{align} 
where $\bs{H}_{\ell}$ is the $n\times K$ matrix of channel coefficients between all users in the network and the base station, $\bs{H}_{-\ell,\ell}$ is the $n\times K_{-\ell}$ matrix of channel coefficients between the $K_{-\ell}$ users outside the cell and the base station, $\bs{X}$ is the $K\times T$ matrix of symbols transmitted by all users in the network, and $\bs{X}_{-\ell}$ is the $K_{-\ell}\times T$ matrix of symbols transmitted by users outside the cell. From the above, it results $K=K_{1}+K_{2}+\cdots+K_{L}$ and $K_{-\ell}=K-K_{\ell}$. For the sake of simplicity, we assume that the columns of $\bs{H}_{\ell}$, which correspond to the channel vectors between the antenna array and a specific user, are independent random variables distributed according to a proper complex Gaussian distribution with diagonal covariance: $\bs{H}_{\ell}\bs{e}_{k}\sim\mathcal{CN}(\bs{0},\bs{D}_{\ell k})$ where $\bs{D}_{\ell k}$ is diagonal, invertible, and known.

Throughout the paper, it is assumed that $n=K$, which corresponds to massive MIMO when $L\gg1$, and that $T>2K$, which we shall refer to as the \textit{long coherence block} assumption. 

Furthermore, we assume that the distribution of transmitted symbols admits a density $\Mathsf{f}$, namely $\bs{X}_{\ell}\sim P$ with $\Mathsf{f}=\mathrm{d}P$, that is independent of all parameters in the network; this excludes, among other consequences, intermittent or bursty activity of users. Moreover, $\bs{X}_{\ell}$ is assumed to be independent of $\bs{X}_{\ell'}$ for all $\ell'\neq \ell$, which models the independence of transmissions in different cells. To summarize,
\be \label{eq:indepsymb} \mathrm{d}P_{\bs{X}}(\bs{X})=\!\bigwedge_{\ell=1}^{L}\Mathsf{f}(\bs{X}_{\ell}) \, \mathrm{d}\bs{X}_{\ell}. \ee
Among these distributions, those such that $h(\bs{X}_{\ell})$ scales as $K_{\ell}T\log\rho$ for large $\rho$ will be referred to as \textit{maximally entropic distributions} and it will be said that they belong to the \textit{maximally entropic ensemble}.

The quantity of interest is the mutual information $I(\bs{X}_{\ell};\bs{Y}_{\ell})$ where symbols are distributed according to \eqref{eq:indepsymb}.

\section{Single-User Setup}
\label{sec:single}

The model presented above reduces to a single-user setup where antennas are used independently when $L=1$. Let the capacity of the single-user MIMO channel be
\[ C\sub{SU} := \sup_{P} \frac{1}{T}I(\bs{X};\bs{Y}) \]
where, according to \eqref{eq:indepsymb}, $\bs{X}\sim P$, $dP(\bs{X})=\Mathsf{f}(\bs{X})d\bs{X}$, and the power constraint is imposed over the rows of $\bs{X}$, namely $\E{}{\| \bs{e}_{i}^{\dag}\bs{X} \|^{2}}\leq \rho T$. In this section, we briefly review the analysis of the degrees of freedom of this channel, the optimality of pilots at high SNR, and finally propose an MLE-based detector. 

\subsection{Capacity and Degrees of Freedom}
A classical result \cite{ZheTse:2002} on the capacity of single-user MIMO systems is that 
\[ C\sub{SU} = n^{*}(1-n^{*}\!/T)\log\rho + o(\log\rho) \]
where $n^{*} = n \wedge K \wedge T/2$. In the context of this paper, where we are assuming $n=K<T/2$, this reduces to
\[  C\sub{SU} = n(1-n/T)\log\rho + o(\log\rho). \]
In other words, the number of degrees of freedom, namely the pre-log factor of capacity at high SNR, over the coherence block is equal to $n(T-n)$.

\subsection{Achieving Capacity With Pilots}
It is possible to achieve $n(T-n)$ degrees of freedom by using $n$ orthogonal pilot sequences. In fact, $n$ out of the $T$ channel uses can be used to transmit the $n$ orthogonal sequences and the remaining $T-n$ channel uses can be used to transmit data. Then, $n$ receive antennas can demultiplex $K=n$ users per channel use. Hence, pilots are asymptotically optimal in the assumed setup.

\subsection{Approaching Capacity Without Pilots}
In this section, we show that the same number of degrees of freedom can be achieved without using pilots. In the single-user setting, this is just another method to approach capacity; however, the extension to cellular networks brings with it nontrivial consequences. 

Denote $\Mathsf{f}_{\theta}\T$ the conditional density of $\bs{Y}$ given a realization of the channel, that is, $\Mathsf{f}_{\bs{Y}|\bs{H}}=:\Mathsf{f}_{\theta}\T$. It will be useful to denote $\bs{Y}^{(k)}$ the received signal $\bs{Y}$ for the realization $\bs{H}=\bs{H}^{(k)}$. For example, the following diagram shows two possible cases:
\[
\hspace{1.5cm}\begin{tikzcd}[ampersand replacement=\&, column sep=4em, row sep=2em]
\bs{X} \ar[dr, "\bs{H}=\bs{H}^{(1)}" description, end anchor = west] \ar[r, "\bs{H}=\bs{H}^{(0)}"] \& \bs{Y}^{(0)}\sim \Mathsf{f}_{\theta_{\hspace{0.5pt}0}}^{\,(T)}:=\Mathsf{f}_{\bs{Y}|\bs{H}=\bs{H}^{(0)}} \\
\& \bs{Y}^{(1)}\sim \Mathsf{f}_{\theta_{1}}^{\,(T)}:=\Mathsf{f}_{\bs{Y}|\bs{H}=\bs{H}^{(1)}}\& 
\end{tikzcd}
\]
The main result of this section is the following:
\begin{theorem}\label{thm:thm1}Let $n=K$. For all absolutely continuous distributions in the maximally entropic ensemble, it holds that
\begin{align*}
 I(\bs{X}; \bs{Y}) & \geqsim h(\bs{Y}|\bs{H})  - nK\log\rho \\
& \simeq n(T-n)\log\rho
\end{align*}
where the conditional differential entropy of the output given the channel can be expressed as
\[ h(\bs{Y}|\bs{H}) = \Epar{}{ - \Epar{}{\log \Mathsf{f}_{\theta^{*}}\T(\bs{Y}) \+ \Big|\+ \bs{H}=\bs{H}^{(0)} } }\]
and $\theta^{*}=\arg\max\limits_{\theta}\bar{\Mathscr{L}}(\theta;\bs{Y}^{(0)})$ with
\begin{align*}
\bar{\Mathscr{L}}(\theta;\bs{Y}) = & - D(\Mathsf{f}_{\theta_{0}}\!\-\T\| \Mathsf{f}_{\theta}\T|\bs{H}=\bs{H}^{(0)}) \\ 
& - \Epar{}{\log \Mathsf{f}_{\theta}\T(\bs{Y}) \+\Big| \+ \bs{H}=\bs{H}^{(0)} }.
\end{align*}
\end{theorem}
\begin{proof}
Let us express the mutual information $I(\bs{X};\bs{Y})$ in terms of the differential entropies $h(\bs{Y})$ and $h(\bs{Y}|\bs{X})$, and study each term separately with the goal of deriving a tight lower bound. The latter conditional entropy can be bounded as follows
\begin{align*}
h(\bs{Y}|\bs{X}) 
& \leqsim n \E{}{\log\det(\bs{I}+\bs{X}\bs{X}^{\dag})}\\ 
& \leqsim n(n\wedge T) \log\rho
 \end{align*}
where the first inequality follows from the independence of channels across antennas, which is a worst-case scenario, and by exploiting their Gaussianity; and the second inequality follows from the input distribution being maximally entropic. %
Overall, this term removes $n^{2}$ degrees of freedom. The former term, $h(\bs{Y})$, can be bounded as follows:
\begin{align*}
h(\bs{Y}) \geqsim h(\bs{Y}|\bs{H}) & \simeq h(\bs{X}) + \E{}{\log\det\bs{H}} \\ & \simeq KT\log\rho
\end{align*}
where the inequality follows from the fact that conditioning reduces differential entropy, and the asymptotic equalities follow from ignoring noise and using the maximal entropic assumption on inputs. Overall, this term acquires $nT$ degrees of freedom. The bound is tight because $h(\bs{Y}) \leq nT\log\rho + O_{\rho}(1)$, and $K=n$ by assumption. The claim follows by expressing $h(\bs{Y}|\bs{H})$ in terms of a Kullback--Leibler divergence:
\begin{align*}
h(\bs{Y}|\bs{H}=\bs{H}^{(0)}) = & -D(\Mathsf{f}_{\theta_{0}}\!\-\T\| \Mathsf{f}_{\theta}\T|\bs{H}=\bs{H}^{(0)}) \\ 
 &  - \Epar{}{\log \Mathsf{f}_{\theta}\T(\bs{Y}) \+\Big|\+\bs{H}=\bs{H}^{(0)} }.
 \end{align*}
Choosing a particular $\theta$, the one maximizing the right hand side, concludes the proof.
\end{proof}
In words, for any maximally entropic input distribution, a detector that estimates the channel by maximizing $\bar{\Mathscr{L}}(\theta;\bs{Y}^{(0)})$ is optimal at high SNR since it exploits all the available degrees of freedom. 

It is important to realize that $\bar{\Mathscr{L}}(\theta;\bs{Y}^{(0)})$ is strictly related to a log-likelihood function, and thus the algorithm above suggests an approach based on maximum likelihood estimation. In fact, upon receiving $\bs{Y}^{(0)}$, the receiver can compute the log-likelihood 
\begin{equation}\label{eq:logli} \Mathscr{L}(\theta;\bs{Y}^{(0)}) := \log \Mathsf{f}_{\theta}\T(\bs{Y}^{(0)}) 
\end{equation}
for some $\theta$. If we had many independent observations of $\bs{Y}^{(0)}$, e.g. a super-channel $\bs{Y}^{(0)}(1), \bs{Y}^{(0)}(2), \ldots, \bs{Y}^{(0)}(N)$, we could compute
\[
\frac{1}{N} \sum_{i=1}^{N} \log \Mathsf{f}_{\theta}\T(\bs{Y}^{(0)}(i\-)) \xrightarrow{\text{a.s.}} \E{}{\Mathscr{L}(\theta;\bs{Y}^{(0)})} = \bar{\Mathscr{L}}(\theta;\bs{Y}^{(0)}).
\]
However, we do not have the luxury of collecting many observations due to the block-fading assumption, and we need to accept the rough estimate $\bar{\Mathscr{L}}(\theta;\bs{Y}^{(0)})\approx \log \Mathsf{f}_{\theta}\T(\bs{Y}^{(0)})$. Thus, in practice, we will solve the following problem:
\be\label{eq:gentheta} \hat{\theta} = \arg\max_{\theta} \log \Mathsf{f}_{\theta}\T(\bs{Y}^{(0)}). \ee
Here, $\theta$ is a parameter that is linked to the channel. An equivalent, more explicit form for \eqref{eq:gentheta} is
\be\label{eq:explicit-dep} \bs{\hat{B}} = \arg\max\limits_{\bs{B}} \Mathscr{L}(\bs{B};\bs{Y}^{(0)}) \ee
where $\Mathscr{L}(\bs{G};\bs{Y}^{(0)}):=\log \Mathsf{f}_{\bs{Y}|\bs{H}=\bs{G}^{-1}}(\bs{Y}^{(0)})$.

Notice though that, in the present context, users can assume that signaling is independent across channel uses without incurring in any loss of degrees of freedom. Under this assumption, the analysis simplifies as follows: the conditional density of the output becomes separable, namely 
\[ \Mathsf{f}_{\theta}\T(\bs{Y}^{(0)}) = \prod_{t=1}^{T} \Mathsf{f}_{\theta}(\bs{y}^{(0)}_{t}) =: \Mathsf{f}_{\theta}^{\otimes T}\!(\bs{y}^{\!(0)}_{1},\dotsc,\bs{y}^{(0)}_{T} ) \]
where $\Mathsf{f}_{\theta}$ is the density over one of the columns $\bs{y}^{(0)}_{1},\dotsc,\bs{y}^{(0)}_{T}$ of $\bs{Y}^{(0)}$; the log-likelihood in \eqref{eq:logli} becomes additive
\[ \log \Mathsf{f}_{\theta}\T(\bs{Y}^{(0)}) = \sum_{t=1}^{T} \log \Mathsf{f}_{\theta}(\bs{y}^{(0)}_{t}); \]
it results $\bar{\Mathscr{L}}(\theta;\bs{Y}^{(0)})=T\bar{\Mathscr{L}}(\theta;\bs{y}_{0})$, $\bs{y}_{0}$ being the generic column of $\bs{Y}^{(0)}$; and one can estimate increasingly better the objective function $\bar{\Mathscr{L}}(\theta;\bs{Y}^{(0)})$ as the coherence block size increases since
\[ \frac{1}{T} \sum_{t=1}^{T} \log \Mathsf{f}_{\theta}(\bs{y}^{(0)}_{t}) - \bar{\Mathscr{L}}(\theta;\bs{y}_{0}) \xrightarrow{\text{a.s.}} 0. \]
The problem in this case is reduced to an \emph{independent component analysis} (ICA) \cite{Com:1994, HyvKarOja:2001}:
\be \hat{\theta} = \arg\max_{\theta} \sum_{t=1}^{T} \log \Mathsf{f}_{\theta}(\bs{y}^{(0)}_{t}). \ee
More explicitly, $\Mathsf{f}_{\theta}$ is given by (we denote $\bs{B}_{0}:= (\bs{H}^{(0)})^{-1}$)
\begin{align*} 
\Mathsf{f}_{\bs{y}|\bs{H} =\bs{H}^{(0)}}(\bs{y}^{(0)}_{t}) 
& = |\det(\bs{H}^{(0)})^{-1}| \,\Mathsf{f}_{\bs{x}}((\bs{H}^{(0)})^{-1}\bs{y}^{(0)}_{t}) \\
& = |\det\bs{B}_{0}| \,\Mathsf{f}_{\bs{x}}(\bs{B}_{0}\bs{y}^{(0)}_{t}). 
\end{align*}
By using the above, we can plug the $T$\!-sample log-likelihood function  
\[ \Mathscr{L}(\bs{B}_{0};\bs{Y}^{(0)}) = T \log|\-\det\bs{B}_{0}| + \sum_{t=1}^{T}\log \Mathsf{f}_{\bs{x}}(\bs{B}_{0}\bs{y}^{(0)}_{t}) \]
into \eqref{eq:explicit-dep} to conclude that
\be\label{eq:explicit} \bs{\hat{B}} = \arg\max\limits_{\bs{B}} \bigg\{\, T \log|\-\det\bs{B}| + \sum_{t=1}^{T}\log \Mathsf{f}_{\bs{x}}(\bs{B}\bs{y}^{(0)}_{t}) \,\bigg\}. \ee


\section{Cellular Setup}
Let the capacity of the uplink in a cell of the noncooperative network be
\[ C\sub{cell} := \sup_{P} \frac{1}{T}I(\bs{X}_{\ell};\bs{Y}_{\ell}) \]
with $P$ satisfying the condition \eqref{eq:indepsymb}. 


\subsection{Upper Bounds}
The following bound is derived by providing to the base station side-knowledge about interference.

\begin{proposition}\label{prop:1}Suppose a genie provides base station $\ell$ with the knowledge of $\bs{H}_{-\ell}\bs{X}_{-\ell}$. Then the following bounds holds:
\[ I(\bs{X}_{\ell};\bs{Y}_{\ell}) \leqsim K_{\ell}(T-K_{\ell})\log\rho.\]
\end{proposition}
\begin{proof}
By providing the side-knowledge of signals outside the cell, we get 
\[ I(\bs{X}_{\ell};\bs{Y}_{\ell}) \leq I(\bs{X}_{\ell};\bs{Y}_{\ell}|\bs{H}_{-\ell}, \bs{X}_{-\ell}) = I(\bs{X}_{\ell};\bs{H}_{\ell}\bs{X}_{\ell}+\bs{Z}).\] 
The bound follows from evaluating the capacity of a cooperative system, which constitues an upper bound on the rightmost mutual information.
\end{proof}
\begin{corollary}As $T\to\infty$, the above bound reduces to
\[ I(\bs{X}_{\ell};\bs{Y}_{\ell}) \leqsim K_{\ell}(T+O_{T}(1))\log\rho. \]
\end{corollary}
\begin{proof}The bound derives directly from Proposition~\ref{prop:1}. Alternatively, it can be derived from a genie-aided detection where both $\bs{H}$ and $\bs{X}_{-\ell}$ are disclosed to the detector.
\end{proof}
The above bound is the maximum mutual information we can achieve since $\bs{X}_{\ell}$ is maximally entropic, namely, $h(\bs{X}_{\ell}) \leqsim TK_{\ell}\log\rho$ under the average power constraint. Notice that both bounds are optimistic in the sense that we are disclosing the interference to the detector.

If the cellular network were cooperative, the overall uplink channel from $(\bs{X}_{1},\dotsc,\bs{X}_{L})$ to $(\bs{Y}_{1},\dotsc,\bs{Y}_{L})$ would have $K(T-K)$ degrees of freedom; thus, on a per-cell basis, it would have $K(T-K)/L$ degrees of freedom. It would be as if cell $\ell$ contributes with $K_{\ell}(T-K)$, which is achievable via orthogonal training over the whole network. 

Together, the above results suggest that, for large $T$, the uplink channel of cell $\ell$ offers $K_{\ell}T$ degrees of freedom, while for smaller $T$ there is a penalty of order $n^{2}$.

\subsection{Approaching Capacity Without Pilots}
The main result of this section is as follows (cf. Theorem~\ref{thm:thm1}).

\begin{theorem}\label{thm:thm2}Let $n=K$. For all absolutely continuous distributions in the maximally entropic ensemble, it holds that
\begin{align}
I(\bs{X}_{\ell}; \bs{Y}_{\ell}) 
& \geqsim h(\bs{Y}_{\ell}|\bs{H}) - nK\log\rho - TK_{-\ell}\log\rho \\
& \simeq  (TK_{\ell}-n^{2})\log\rho\label{eq:thm2-2}
\end{align}
where the conditional differential entropy of the output given the channel can be expressed as
\[ h(\bs{Y}_{\ell}|\bs{H}) = \Epar{}{ - \Epar{}{\log \Mathsf{f}_{\theta^{*}}\T(\bs{Y}_{\ell}) \+ \Big|\+ \bs{H}=\bs{H}^{(0)}_{\ell} } }\]
and $\theta^{*}=\arg\max\limits_{\theta}\bar{\Mathscr{L}}(\theta;\bs{Y}_{0\ell})$ with
\begin{align*}
\bar{\Mathscr{L}}(\theta;\bs{Y}_{\ell}) = & - D(\Mathsf{f}_{\theta_{0}}\!\-\T\| \Mathsf{f}_{\theta}\T|\bs{H}=\bs{H}^{(0)}_{\ell}) \\ 
& - \Epar{}{\log \Mathsf{f}_{\theta}\T(\bs{Y}_{\ell}) \+\Big| \+ \bs{H}=\bs{H}^{(0)}_{\ell} }.
\end{align*}
\end{theorem}
\begin{proof} Split the mutual information $I(\bs{X}_{\ell};\bs{Y}_{\ell})$ as follows:
\[ I(\bs{X}_{\ell};\bs{Y}_{\!\ell}) = I(\bs{X}; \bs{Y}_{\!\ell}) - {I(\bs{X}_{-\ell};\bs{Y}_{\!\ell}|\bs{X}_{\ell})}. \] 
The first term on the right-hand side can be treated as in the single-user case. Specifically, let analyze the mutual information $I(\bs{X}; \bs{Y}_{\ell})$ in terms of $h(\bs{Y}_{\ell})$ and $h(\bs{Y}_{\ell}|\bs{X})$. The first term, $h(\bs{Y}_{\ell})$, can be rewritten as follows
\[ h(\bs{Y}_{\ell}) = \Epar{}{ - \Epar{}{\log \Mathsf{f}_{\theta^{*}}\T(\bs{Y}_{\ell}) \Big|\bs{H}=\bs{H}^{(0)}_{\ell} } } \]
where $\theta^{*}$ is defined in the statement. The second term, $h(\bs{Y}_{\ell}|\bs{X})$, can be bounded by noticing that the conditional distribution of $\bs{Y}_{\ell}$ given $\bs{X}$ is Gaussian, and thus 
\begin{align*}
h(\bs{Y}_{\ell}|\bs{X}) 
& \leqsim n \E{}{\log\det(\bs{X}\bs{X}^{\dag})} \\
& \leqsim n (K\wedge T)\log\rho,
\end{align*}
where we ignored channel attenuations, which do not play any role asymptotically at high SNR. 

Therefore, in summary, the first mutual information is tightly bounded as follows:
\begin{align}
\hspace{-1ex}I(\bs{X}; \bs{Y}_{\ell}) 
&\, {\geqsim \Epar{}{ - \Epar{}{\log \Mathsf{f}_{\theta^{*}}\T(\bs{Y}_{\ell}) \Big|\bs{H}=\bs{H}^{(0)}_{\ell} } }} - n^{2}\log\rho,\label{eq:lb}
\end{align}
which follows from using side-information provided by a genie about all channels.

The second mutual information, $I(\bs{X}_{-\ell};\bs{Y}_{\!\ell}|\bs{X}_{\ell})$, can be upper bounded as follows:
\begin{align}
I(\bs{X}_{-\ell};\bs{Y}_{\ell}|\bs{X}_{\ell}) 
& \leq I(\bs{X}_{-\ell};\bs{Y}_{\ell}|\bs{H}, \bs{X}_{\ell}) \nonumber \\ 
& = I(\bs{X}_{-\ell};\bs{H}_{-\ell}\bs{X}_{-\ell}+\bs{Z}|\bs{H}_{-\ell}) \nonumber\\
& \leq T \E{}{\log\det(\bs{I}+\rho\bs{H}_{-\ell}^{\phantom{\dag}}\bs{H}_{-\ell}^{\dag})} \nonumber \\
& \leqsim T K_{-\ell}\log \rho \label{eq:ub2}
\end{align}

Using together \eqref{eq:lb} and \eqref{eq:ub2} concludes the proof.
\end{proof}

Another way to interpret Theorem~\ref{thm:thm2} is as follows:
\[  \sup  I(\bs{X}_{\ell}; \bs{Y}_{\ell}) \geqsim [TK_{\ell} + O_{T}(1)]\log\rho. \]
That is, the pre-log factor loss is constant in $T$, which suggests that, for long coherence blocks, the bound is tight. In practice, \eqref{eq:thm2-2} suggests that $TK_{\ell}\gg n^{2}$ for the bound to be close to the coherent capacity, which is equivalent to requiring $T\gg Ln$.

\section{Simulations}
We provide simulations assuming independent symbols over the channel uses, which is a setup similar to ICA. For the sake of simplicity, symbols and channels are real. Symbols are distributed according to a zero-mean, unit-variance Laplacian distribution. The $T$\!-sample log-likelihood for Laplacian sources takes the form
\be\label{eq:ll-laplacian}
\Mathscr{L}(\bs{B};\bs{Y}) = T\log|\det\bs{B}| - \sqrt{2} \sum_{t=1}^{T} \sum_{k=1}^{K}|(\bs{B}\bs{Y})_{kt}|.
\ee
The specific form of the channels is not relevant as long as they are linearly independent. For the purpose of this section, all channels are drawn from an i.i.d. Gaussian distribution with variance accounting for the user-depedent attenuation. Attenuations are drawn at random from a uniform distribution with support $[0.1, 1.9]$. %
Once drawn, the channels are fixed over the coherence block of size $T$.

In Figure~\ref{fig:1} and \ref{fig:4}, we show the absolute value of the statistical correlation coefficient between estimated and true symbols:%
\be \rho_{ij} = \frac{\bs{e}_{i}^{\dag}\bs{X}\bs{\hat{X}}^{\dag}\bs{e}_{j}}{\|\bs{e}_{i}^{\dag}\bs{X}\| \, \|\bs{e}_{j}^{\dag}\bs{\hat{X}}\|}. \ee
Here, $\bs{e}_{i}^{\dag}\bs{X}$ denotes the symbols transmitted by user $i$ over the coherence block. Thus, $\rho_{ij}$ represents the normalized inner product of transmitted and estimated data, which can be interpreted as their statistical correlation coefficient. Since the true signals are independent, we expect $|\rho_{ij}|\approx 1$ for only one $j$ for any fixed $i$; in other words, we expect $(|\rho_{ij}|)$ to be close to a permutation matrix. The order of users in $\bs{X}$ and $\bs{\hat{X}}$ is not the same due to the symmetries of the likelihood function. 

\begin{figure}[tb]
\hspace{-2ex}\includegraphics[width=0.55\textwidth]{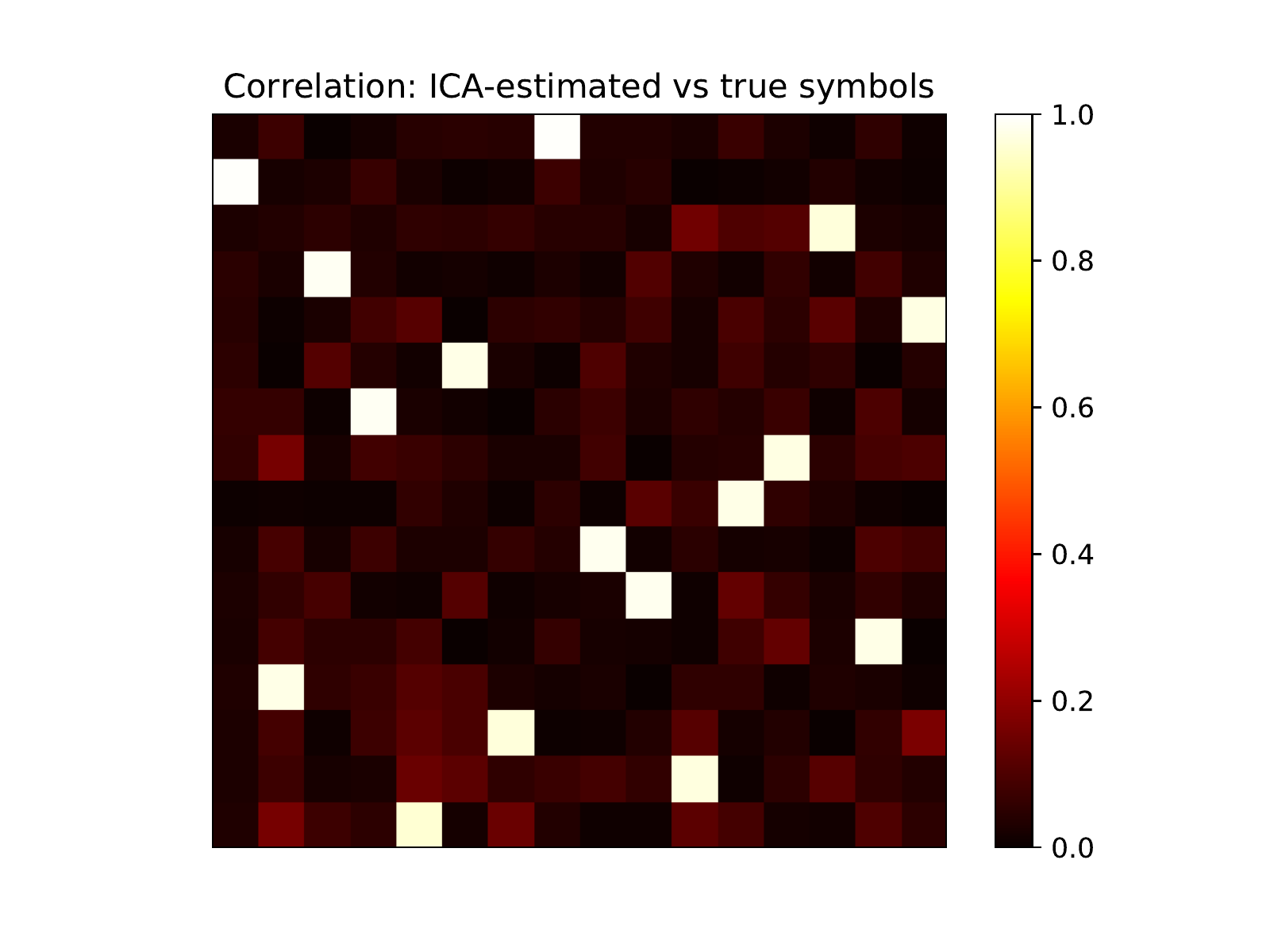}\vspace{-4mm}
\caption{Absolute value of the statistical correlation coefficient between $\bs{\hat{X}}$ and $\bs{X}$ for one realization only of the channel. Parameters: $n=K=16$, $T=2nK=2n^{2}=512$.}
\label{fig:1}
\vspace{-3mm}
\end{figure}

\begin{figure}[tb]
\hspace{-2ex}\includegraphics[width=0.55\textwidth]{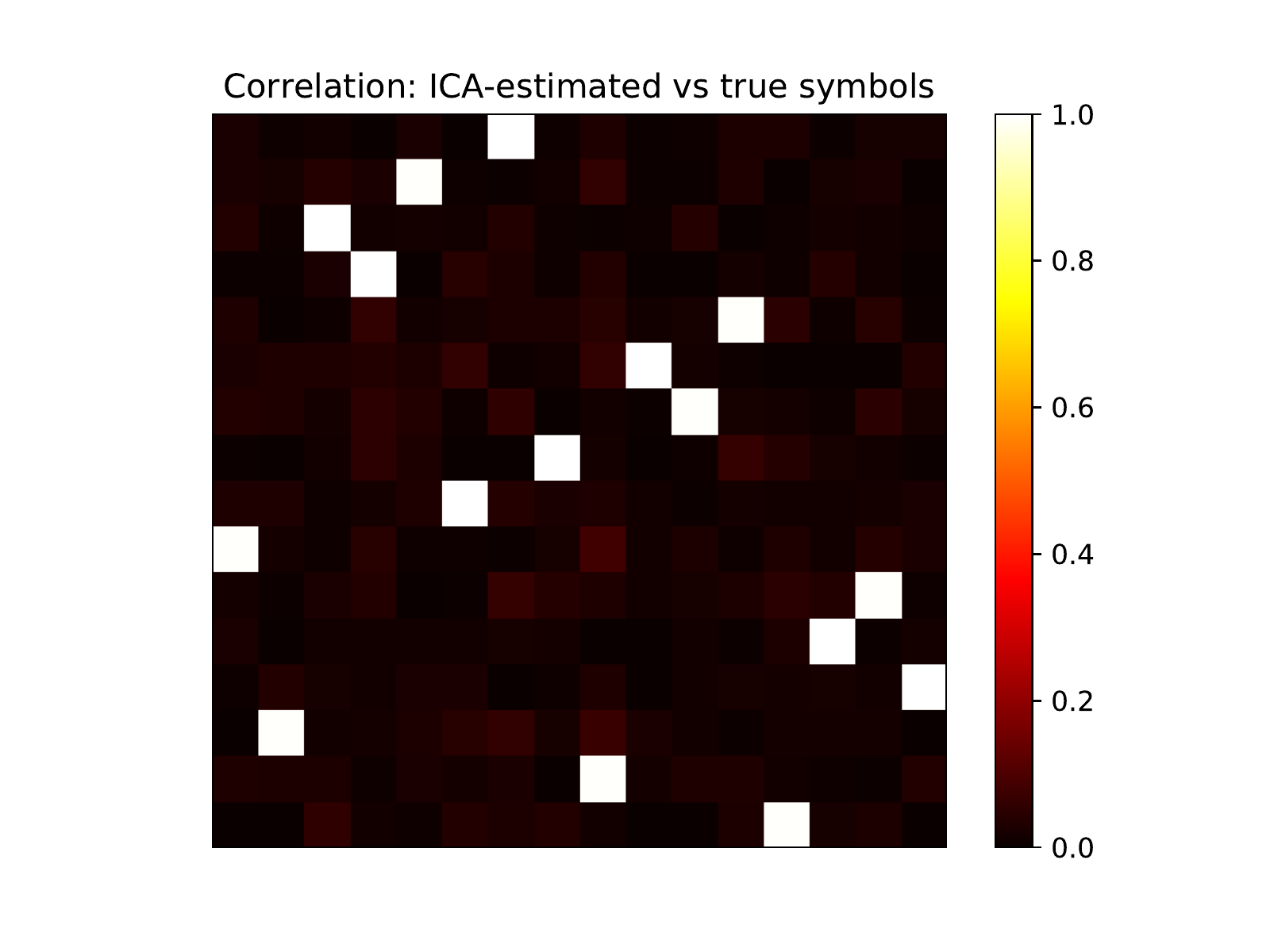}\vspace{-4mm}
\caption{Absolute value of the statistical correlation coefficient between $\bs{\hat{X}}$ and $\bs{X}$ for one realization only of the channel. Parameters: $n=K=16$, $T=8 n^{2}=2048$.}
\label{fig:4}
\vspace{-3mm}
\end{figure}

The two figures show the same setting, in terms of number of users and antennas, with different coherence block sizes. In both cases, we get a picture very close to a permutation matrix. In Figure~\ref{fig:4}, the coherence block length is larger, and the performance is better as expected. We can also observe that the permutation between the two figures is different, which derives from the optimizer following different paths.

\begin{figure}[tb]
\hspace{-2ex}\includegraphics[width=0.55\textwidth]{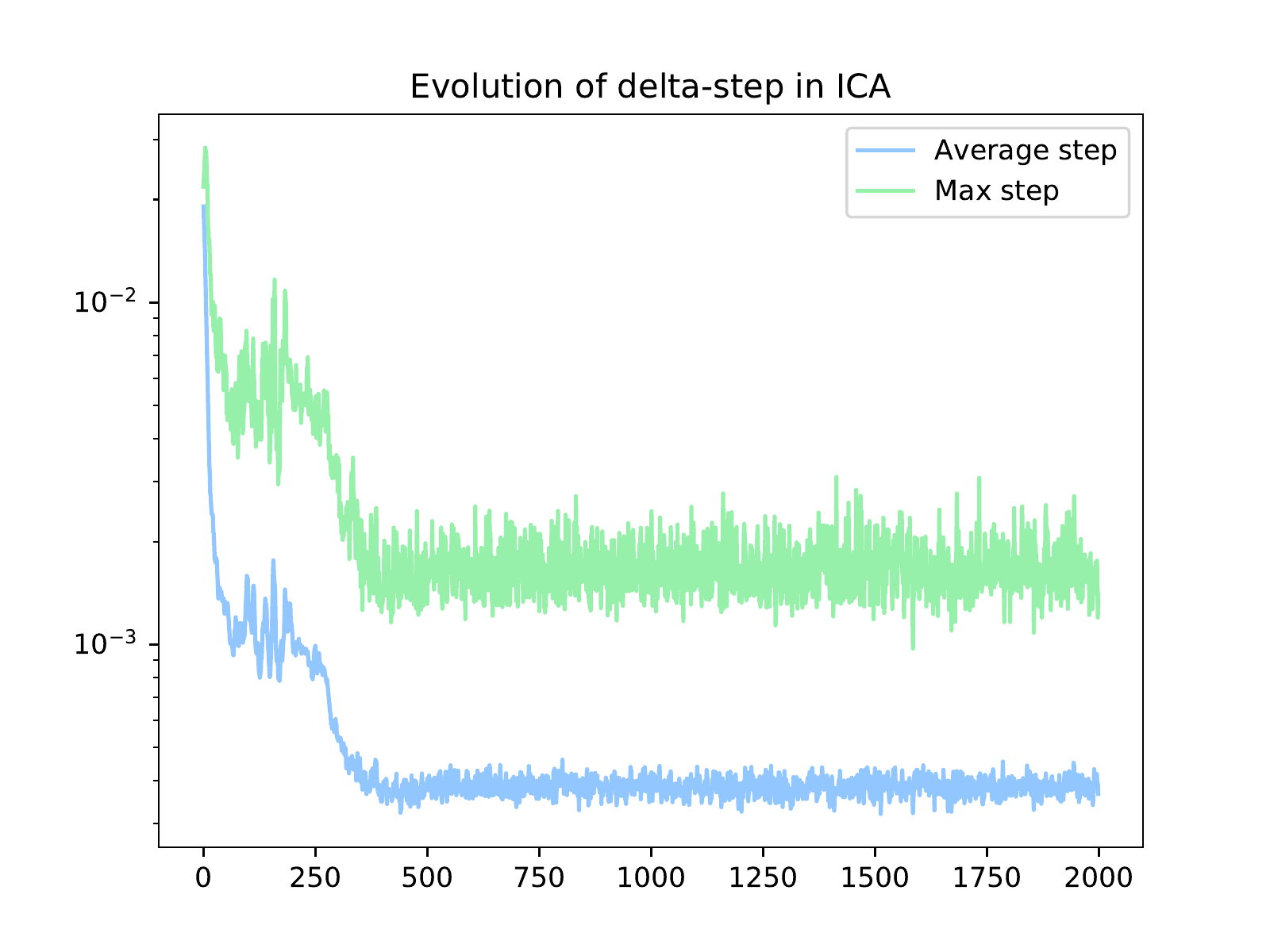}\vspace{-2mm}
\caption{Parameters: $n=K=16$, $T=8 n^{2}=2048$. The x-axis shows the number of steps made during the optimization, each step accounting for 10 iterations.}
\label{fig:2}
\end{figure}

In order to maximize \eqref{eq:ll-laplacian} with respect to $\bs{B}$, we leverage the optimizers implemented in Tensorflow \cite{TF:2016, TFweb:2019}. For the specific figures in this paper, Adam optimizer was used \cite{KinBa:2015}; similar results are achieved with stochastic gradient descent (SGD). The optimizer is initialized at random, e.g. $\bs{B}$ at iteration $0$, denoted $\bs{B}^{[0]}$, is drawn from a Gaussian ensemble, and run for a fixed number of iterations. At each iteration, the entries of the matrix are slightly changed: $\bs{B}^{[k]}=\bs{B}^{[k-1]} + \Delta \bs{B}^{[k-1]}$. Figure~\ref{fig:2} shows the evolution of two quantities derived from $\Delta \bs{B}^{[k-1]}$ as a function of the iteration $k$, namely the \textit{average step} 
$\frac{1}{n^{2}}\sum_{i,j}|B_{ij}^{[k]}-B_{ij}^{[k-1]}|$,  
where $B_{ij}^{[k]}$ is the value of the element $(i,j)$ in $\bs{B}$ at iteration $k$, and the \textit{maximum step}, $\max_{ij} |B_{ij}^{[k]}-B_{ij}^{[k-1]}|$.
The objective function is highly nonlinear, and there is no guarantee of convergence. However, in our numerical experiments, we always found a very good local maximum the log-likelihood, and both average and maximum steps tend to decrease as the number of iterations grows, which is a good indication of convergence.

\bibliographystyle{IEEEtran}
\bibliography{IEEEabrv,publishers,confs-jrnls,refs,biblio-clean}




\end{document}